\documentclass[sigconf,nonacm]{acmart}
\AtBeginDocument{%
  }

\setcopyright{none}

\usepackage{multirow}
\usepackage{colortbl}
\usepackage{tikz}
\usepackage{placeins}
\usepackage[normalem]{ulem}
\usepackage{colortbl}
\newcommand{\csragrowstretch}{0.90}
\AtBeginEnvironment{tabular}{}
\AtBeginEnvironment{tabular*}{}
\newlength{\csragOriginalTextFloatSep}
\newlength{\csragOriginalDblTextFloatSep}
\newlength{\csragOriginalFloatSep}
\newlength{\csragOriginalDblFloatSep}
\newlength{\csragOriginalInTextSep}
\begin{document}

\title{Don't Worry About Quality: Robust Multi-Hop Retrieval despite Issues in Knowledge Graphs}

\author{Yizhuo Ma\texorpdfstring{\textsuperscript{1,*}}{}}
\email{myz@std.uestc.edu.cn}
\affiliation{%
  \institution{University of Electronic Science and Technology of China}
  \country{China}}
\author{Ke Qin\texorpdfstring{\textsuperscript{1}}{}}
\affiliation{%
  \institution{University of Electronic Science and Technology of China}
  \country{China}}
\author{Jinchuan Xu\texorpdfstring{\textsuperscript{1}}{}}
\affiliation{%
  \institution{University of Electronic Science and Technology of China}
  \country{China}}
\author{Tao Wen\texorpdfstring{\textsuperscript{1}}{}}
\affiliation{%
  \institution{University of Electronic Science and Technology of China}
  \country{China}}
\author{Rongzheng Wang\texorpdfstring{\textsuperscript{1}}{}}
\affiliation{%
  \institution{University of Electronic Science and Technology of China}
  \country{China}}
\author{Jiakai Li\texorpdfstring{\textsuperscript{1}}{}}
\affiliation{%
  \institution{University of Electronic Science and Technology of China}
  \country{China}}
\author{Nan Chen\texorpdfstring{\textsuperscript{2}}{}}
\affiliation{%
  \institution{Department of Thoracic Surgery, West China Hospital, Sichuan University}
  \country{China}}
\author{Lunxu Liu\texorpdfstring{\textsuperscript{2}}{}}
\affiliation{%
  \institution{Department of Thoracic Surgery, West China Hospital, Sichuan University}
  \country{China}}
\author{Deqiang Ouyang\texorpdfstring{\textsuperscript{3}}{}}
\affiliation{%
  \institution{Chongqing University}
  \country{China}}
\author{Shuang Liang\texorpdfstring{\textsuperscript{1}}{}}
\affiliation{%
  \institution{University of Electronic Science and Technology of China}
  \country{China}}
\author{Man-Fai Leung\texorpdfstring{\textsuperscript{4}}{}}
\affiliation{%
  \institution{Anglia Ruskin University}
  \country{United Kingdom}}
\renewcommand{\shortauthors}{Ma et al.}

\newcount\csragAuthorIndex
\newcommand{\csragAuthorNames}{%
  \begingroup
  \csragAuthorIndex=0
  \def\and{%
    \advance\csragAuthorIndex by 1
    \ifnum\csragAuthorIndex=6 \\[2pt]\else\quad\fi}%
  \authors
  \endgroup}
\makeatletter
\renewcommand{\@mkauthors}{%
  \global\setbox\mktitle@bx=\vbox{%
    \unvbox\mktitle@bx\par\medskip
    \centering
    {\normalsize\csragAuthorNames\par}
    \vspace{4pt}
    {\small
      \textsuperscript{1}University of Electronic Science and Technology of China\par
      \textsuperscript{2}Department of Thoracic Surgery, West China Hospital, Sichuan University\par
      \textsuperscript{3}Chongqing University\quad
      \textsuperscript{4}Anglia Ruskin University\par
    }
    \medskip}}
\makeatother

\begin{abstract}

Graph Retrieval-Augmented Generation (GraphRAG) has made substantial progress in multi-hop reasoning, but typically assumes that constructed knowledge graphs (KGs) provide sufficient structural support. In practice, cost and computing constraints often leave constructed KGs imperfect, raising the key challenge of achieving robust multi-hop retrieval. To address this challenge, we empirically analyze KG extraction across builders and find similar error patterns despite differences in error proportions, motivating a retrieval strategy targeted at these shared errors. We further summarize these patterns into two issue modes, spurious noise and incomplete information, which induce retrieval drift and retrieval hallucination by misleading evidence selection and weakening evidential support, respectively. Based on this, we propose CS-RAG, a robust GraphRAG framework that mitigates the impact of imperfect KGs during retrieval. CS-RAG first plans each query as an ordered sequence of executable atomic constraints and performs fine-grained anchor- and relation-aware retrieval to constrain evidence acquisition around the intended hop semantics. It then applies a sufficiency check to decide whether retrieved entities should serve as anchors for subsequent retrieval hops and activates textual recovery for unresolved constraints, thereby reducing hallucinated structural continuation. Experiments on three multi-hop QA benchmarks show that CS-RAG is less sensitive to builder choice and remains stable under controlled KG issue injection.
Code is available at: \url{https://github.com/myz12138/CS-RAG}
\end{abstract}



\maketitle
\begingroup
\renewcommand{\thefootnote}{*}
\footnotetext[0]{Email: \href{mailto:myz@std.uestc.edu.cn}{\texttt{myz@std.uestc.edu.cn}}.}
\endgroup

\section{Introduction}

Retrieval-Augmented Generation (RAG)~\cite{rag1} grounds LLM~\cite{qwen3_technical_report_2025,meta2024llama33_modelcard,openai2025gpt5systemcard} outputs in external documents, but retrieval becomes brittle when a question requires several passages to be linked in a specific order. Graph Retrieval-Augmented Generation (GraphRAG)~\cite{graphrag_survey1,graphrag_survey2} addresses this limitation by using LLM-constructed knowledge graphs (KGs)~\cite{KG1,KG2,KG3} as structured retrieval indexes, making relations among evidence units explicit for multi-hop reasoning~\cite{rag_app2,RAG_APP1}. Recent methods exploit such graph structure through global propagation, learned graph scoring, or path-centric pruning to improve evidence expansion and selection, yielding promising results~\cite{hipporag,gfm-rag,pathrag,neuropath}.

\begin{figure*}[!t]
\centering
\includegraphics[width=0.985\linewidth]{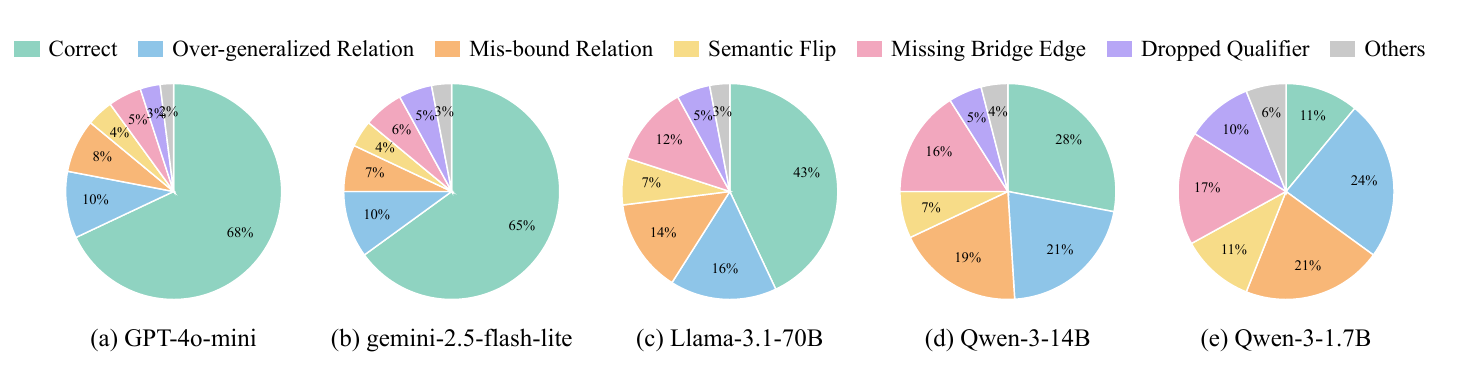}
\caption{Distribution of correct extractions and incorrect extraction patterns under different KG builders. The statistics are computed from MuSiQue, 2WikiMultihopQA, and HotpotQA, with 200 queries sampled from each dataset.}
\label{fig:observation}
\end{figure*}

Despite their different retrieval designs, these methods largely share one implicit assumption that the constructed KG is reliable enough to serve as a structural index.
In practice, the choice of LLM builders for KG construction~\cite{builder1,builder2} is constrained by cost and computing resources. Extracting KGs from large collections of documents under a limited budget may necessitate lower-cost models. Local processing requirements for sensitive documents may also restrict model size to what available hardware supports. Deployments therefore require builders suited to their constraints, including lightweight open-source models that can run locally. However, differences in extraction capability across builders~\cite{builder3} lead to varying degrees of KG quality issues, while the additional cost of verifying and correcting extracted facts limits the feasibility of direct KG repair. These constraints raise a key challenge: \textit{How can we make retrieval robust to variations in KG quality?}

To address this challenge, we first conduct an empirical analysis of KG quality differences across builders. Specifically, we examine the tuples extracted by each builder from query-relevant gold supporting sentences across datasets and quantify the distributions of correct extractions and incorrect extraction patterns, as shown in Figure~\ref{fig:observation}. We find that builders with different extraction capabilities exhibit similar error patterns, although their relative proportions vary. (Appendix~\ref{app:issue_patterns} provides detailed descriptions.) We further analyze their effects on downstream retrieval and summarize these patterns into two broader issue modes. (1) \textbf{Spurious noise}. Over-generalized relations, mis-bound relations, and semantic flips preserve plausible connectivity while weakening, misassigning, or reversing the intended relation semantics. Such semantically distorted yet structurally plausible evidence readily misleads retrieval, giving rise to \textbf{retrieval drift}. (2) \textbf{Incomplete information}. Missing bridge edges and dropped qualifiers deprive evidence chains of necessary connections and constraints, including temporal and spatial conditions. These gaps leave subsequent retrieval steps without sufficient evidential support, thereby inducing \textbf{retrieval hallucination}. These shared error patterns motivate a targeted strategy for mitigating their impact on retrieval under varying levels of KG quality. We elaborate on this analysis in Section~\ref{sec:motivation}.

Building on this, we propose \textbf{CS-RAG} (\underline{C}onstraint-based and \underline{S}ufficiency-guided \underline{R}etrieval \underline{A}ugmented \underline{G}eneration), a robust framework for multi-hop retrieval over imperfect KGs.
CS-RAG couples constraint-based evidence selection with explicit control over whether candidate bindings can anchor subsequent retrieval, mitigating retrieval drift and hallucination across hops.
\textbf{First}, the constraint-based retrieval module represents a query as atomic constraints with explicit intermediate variables and relation variants (e.g., \textit{Apollo 13} -- \{\textit{nominated for}, \ldots, \textit{up for}\} -- ?award). Each constraint guides anchor matching, relation alignment, and contextual reranking, focusing evidence selection on the relation required at that hop to reduce interference from structurally plausible but semantically misaligned triples. \textbf{Second}, the sufficiency-guided propagation and recovery module introduces a hop-wise sufficiency criterion into multi-hop retrieval. Our analysis shows that this criterion yields a lower bound on the probability mass assigned to the top-ranked candidate, supporting its use in limiting propagation from ambiguous candidate pools. Specifically, CS-RAG performs a distribution-based check to decide whether the retrieved candidates provide reliable support for structural binding or whether textual recovery should instead be triggered. Our contributions are as follows:

 1. We empirically identify similar KG extraction error patterns across builders despite differences in extraction capability and error proportions. We further analyze how these patterns affect downstream retrieval and summarize them into two issue modes, spurious noise and incomplete information, which induce retrieval drift and retrieval hallucination, respectively.

 2. We propose CS-RAG to achieve robust retrieval against KG issues. It uses atomic constraints to mitigate retrieval drift, while sufficiency checks guide propagation and trigger textual recovery to suppress retrieval hallucination.

 3. Experiments on three multi-hop QA benchmarks demonstrate stronger robustness across KG builders and under controlled KG issue injection, with CS-RAG maintaining stable QA performance despite variations in KG quality.

\section{Related Work}

\subsection{Graph-structured Retrieval}

Early graph-structured retrieval methods, such as GraphRAG~\cite{graphrag1} and HippoRAG~\cite{hipporag}, propagate query relevance over graph indexes through global community detection or Personalized PageRank~\cite{page1999pagerank}, while LightRAG~\cite{lightrag} further combines unstructured vector retrieval with structured neighborhood expansion to cover dispersed evidence. More recent methods introduce learned structural modeling, where GNN-RAG~\cite{GNN-RAG}, G-Retriever~\cite{G-retrieval}, and SubgraphRAG~\cite{subgraphrag} use graph neural networks to score nodes or subgraphs, and GFM-RAG~\cite{gfm-rag} learns query-dependent structural representations to support zero-shot retrieval across graph domains.  Despite differences in graph representation, scoring, and propagation, incorrect connections can still steer retrieval toward unsupported evidence, while missing connections can break evidence chains and leave later hops under-supported.

\subsection{Iterative Multi-hop Retrieval}

To control the search complexity of multi-hop retrieval, another line of work progressively constructs retrieval paths through iterative selection and controlled expansion. Recent methods formulate retrieval as sequential path construction, using LLM-driven decisions or graph-guided heuristics to prune branches and retain relevant trajectories~\cite{tog,tog2}. Other approaches further incorporate symbolic query solving or structure-guided planning to align intermediate steps with graph structure~\cite{kag,youturag}. Meanwhile, NeuroPath~\cite{neuropath} uses LLM-guided path validation and expansion decisions to construct reasoning paths, followed by supplementary text retrieval to enrich the evidence. PathRAG~\cite{pathrag} employs flow-based allocation and adaptive expansion to balance search breadth and depth.  These methods constrain multi-hop search through candidate relevance or search policies, but provide limited explicit control over whether bindings from ambiguous candidate pools should propagate to subsequent hops. CRAG~\cite{crag} uses retrieved-document quality to guide further retrieval, but does not extend this control to graph-based retrieval. As a result, when the structural support is insufficient, retrieval may still continue from available candidates, allowing uncertain intermediate results to propagate and gradually form an insufficiently grounded retrieval chain.

\section{Methodology}
In the following, we first explain the motivation for addressing retrieval drift and retrieval hallucination. Building on this analysis, we then present the CS-RAG framework and the detailed design of its core components.

\subsection{Motivation}
\label{sec:motivation}
\noindent\textbf{\textit{How to Address Retrieval Drift.}}
Spurious noise often preserves plausible graph connectivity and semantic relevance to the query, making erroneous evidence appear credible. In both global graph retrieval~\cite{hipporag} and hop-wise path search~\cite{tog}, relevance signals alone may not adequately distinguish such evidence from evidence that supports the current binding. Undetected local semantic deviations propagate through erroneous intermediate bindings, directing subsequent retrieval toward incorrect entities and causing retrieval drift. We thus propose narrowing broad relevance-based search to candidate selection under explicit constraints, with shared variables encoding dependencies among constraints. Entity anchors, variable types, and relation variants impose fine-grained restrictions on the candidate space to reduce interference from spurious noise during evidence selection. Retrieval then proceeds along these dependencies to form a compact evidence chain.

\noindent\textbf{\textit{How to Address Retrieval Hallucination.}}
Incomplete information leaves some retrieval hops without the evidence needed to support entity bindings, yet relevance ranking still selects the best matches among available candidates. When the candidate distribution lacks concentration, the reranker scores provide limited separation among candidates.Selecting a top-ranked candidate then carries unresolved ambiguity into subsequent retrieval. Continued expansion from such bindings turns the original evidence gap into a seemingly complete structural chain, resulting in retrieval hallucination. We therefore propose separating constraint-guided evidence selection from binding propagation, using candidate concentration as an empirical signal to decide whether retrieved bindings should anchor subsequent hops. Constraints that fail the check remain unresolved and receive textual evidence without inducing further structural bindings.

\subsection{Overview}
Building on this analysis, we develop CS-RAG, a framework that addresses retrieval drift and retrieval hallucination for robust multi-hop retrieval.
Given a query and a constructed KG, CS-RAG first builds an ordered constraint plan with explicit relation semantics and intermediate variables, and then performs constraint-based retrieval for executable constraints. Each candidate pool is checked before its bindings can anchor later constraints, while unresolved constraints receive textual evidence without further binding propagation. This yields constraint-aligned context for LLM generation. Figure~\ref{fig:framework} illustrates the workflow of CS-RAG.

\begin{figure*}[!t]
    \centering
    \includegraphics[width=0.95\linewidth]{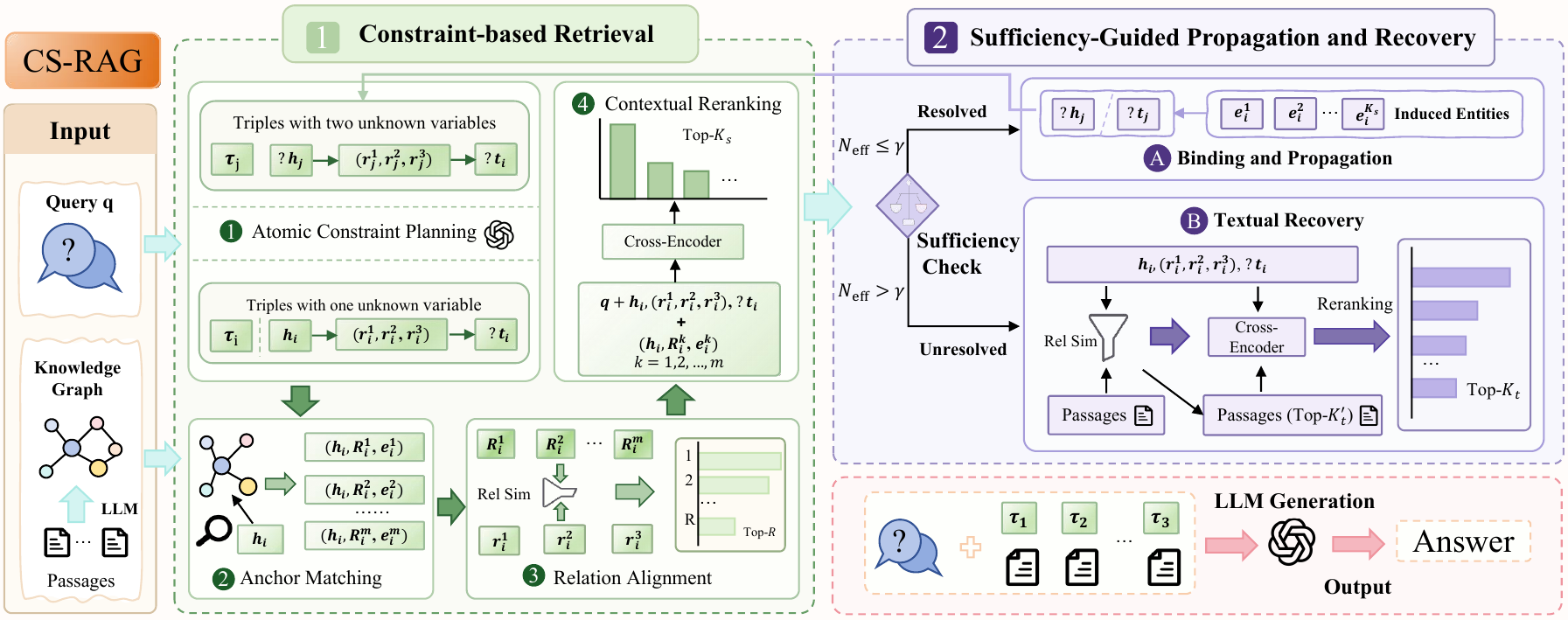} 
    \caption{The Overview of CS-RAG workflow.}
    \label{fig:framework}
\end{figure*}

\subsection{Constraint-based Retrieval}
We convert a multi-hop query into an ordered plan of atomic constraints with relation variants. Evidence is then retrieved for each constraint according to its endpoints.

\noindent\textbf{Atomic Constraint Planning.}
Given a query $q$, we use an LLM-based planner to produce a structured plan
$\mathcal{P}=\langle \tau_1,\tau_2,\ldots,\tau_k\rangle$.
Each $\tau_i=(h_i,r_i,t_i)$ represents an atomic constraint in the form of a triple pattern. Each endpoint belongs to either the observed entity set $\mathcal{E}_q$ extracted from $q$ or the unknown-variable set $\mathcal{V}=\{?v_1,?v_2,\ldots\}$. Each variable is represented by a typed “\textit{?}” placeholder that denotes an intermediate entity slot and may be shared across constraints.
For each $\tau_i$, the planner also generates $m$ relation variants $\mathcal{R}_i=\{r_i^{(1)},\ldots,r_i^{(m)}\}$. By listing plausible expressions of the relation implied by $q$, these variants make the intended relation explicit and provide targets for matching and filtering KG relations during constraint-based retrieval.
We then partition the constraints in $\mathcal{P}$ into $\mathcal{P}_{\mathrm{var}}$, with at least one variable endpoint, and $\mathcal{P}_{\mathrm{obs}}$, with two observed entity endpoints.

\noindent\textbf{Anchor Matching.}
Within $\mathcal{P}_{\mathrm{var}}$, let $\mathcal{P}_1$ and $\mathcal{P}_2$ denote constraints with one and two variable endpoints, respectively. We first retrieve candidates for constraints in $\mathcal{P}_1$ using their observed endpoints as anchors.
For each $\tau_i\in\mathcal{P}_1$, we denote the constant endpoint as the anchor surface form $e_i^{surf}\in\mathcal{E}_q$,
and treat the other endpoint as the unique variable whose assignment will be induced by retrieval.

Let the constructed KG be $\mathcal{G}=(\mathcal{E},\mathcal{R},\mathcal{T})$,
where $\mathcal{T}\subseteq\mathcal{E}\times\mathcal{R}\times\mathcal{E}$ is the triple set.
We map $e_i^{surf}$ to anchor nodes \(\mathcal{E}^{anchor}_i \subset \mathcal{E}\) using top-3 token-coverage matching between the query entity mention and KG entity names. 
Starting from $\mathcal{E}^{anchor}_i$, we enumerate the one-hop neighborhood to form the candidate pool $\mathcal{C}_i$.
This localized enumeration restricts retrieval to neighborhoods directly connected to the anchor, while preserving candidates that can instantiate the unique variable in $\tau_i$.

\noindent\textbf{Relation Alignment and Contextual Reranking.}
For each candidate triple $c\in\mathcal{C}_i$ associated with $\tau_i$,
we first compare its KG relation against the relation variants in $\mathcal{R}_i$ and retain the top-$R$ aligned candidates to obtain $\mathcal{C}_i^{init}$.
This step focuses retrieval on relations aligned with the intended constraint.
We then apply a cross-encoder reranker for contextual scoring.
For each $c\in\mathcal{C}_i^{init}$, we score the concatenated input of the query, the linearized constraint, and the linearized candidate triple, and keep the top-$K_s$ candidates as the refined pool $\mathcal{C}_i^{top}$.
This step yields a compact candidate pool that favors alignment with the current atomic constraint over broad semantic relevance.

\subsection{Sufficiency-Guided Propagation and Recovery}
This stage checks each refined candidate pool before its entities can serve as anchors for dependent constraints. Passing pools provide alternative variable bindings for consolidation and propagation. Unresolved constraints receive textual evidence without inducing further bindings.

\noindent\textbf{Sufficiency Scoring.}
For each atomic constraint $\tau_i\in\mathcal{P}_1$, the previous step returns a refined candidate pool $\mathcal{C}_i^{top}$
with cross-encoder scores $\{z_{i,c}\}_{c\in\mathcal{C}_i^{top}}$.
Empty pools are unresolved and trigger textual recovery. Otherwise, we normalize these scores with softmax into a probability distribution $\{p_{i,c}\}$ over the candidates, and compute the effective number of candidates:
\begin{equation}
N_{\mathrm{eff}}(\tau_i)=\frac{1}{\sum_{c\in\mathcal{C}_i^{top}}p_{i,c}^2}.
\label{eq:neff}
\end{equation}
$N_{\mathrm{eff}}(\tau_i)$ expresses the candidate pool as an equivalent number of equally weighted candidates. A smaller value reflects a more pronounced preference among candidates, while a larger value indicates limited score separation. We use this distinction as an empirical signal for deciding whether bindings from the pool should guide subsequent retrieval.
We use a threshold $\gamma$ to decide whether to propagate bindings from the candidate pool. Appendix~\ref{app:neff} characterizes the concentration imposed by this threshold.
\begin{equation}
\mathrm{State}(\tau_i) =
\begin{cases}
\mathrm{Resolved},   & N_{\mathrm{eff}}(\tau_i) \le \gamma,\\
\mathrm{Unresolved}, & N_{\mathrm{eff}}(\tau_i) > \gamma.
\end{cases}
\label{eq:state}
\end{equation}
Among variable-bearing constraints processed by structural retrieval, we denote the resolved and unresolved subsets by $\mathcal{P}_R$ and $\mathcal{P}_U$, respectively, with constraints having empty candidate pools included in $\mathcal{P}_U$.

\noindent\textbf{Textual Recovery for Unresolved Constraints.}
For an unresolved constraint $\tau_u\in\mathcal{P}_U$, we activate textual recovery to supplement evidence without committing uncertain structural bindings.
Given a query-specific set of passages $U(q)$, we first retrieve the top-$K_t'$ ($K_t'$=10) passages by measuring semantic similarity between the linearized constraint and each passage. 
We then rerank the retrieved passages with a cross-encoder over the query, the current constraint, and the passage text, and keep the top-$K_t$ passages as the supplemental textual evidence set $\mathcal{T}_u^{txt}$.
This recovery branch reduces propagating uncertain structural bindings.

\noindent\textbf{Binding and Propagation for Resolved Constraints.}
Passing the sufficiency check indicates the candidate pool offers a viable starting point for subsequent retrieval. We retain multiple candidate entities to reduce the risk of excluding a valid binding by selecting only the highest-scoring candidate.
Specifically, each resolved constraint $\tau_i\in\mathcal{P}_R$ has one unbound variable $v$ when it is processed.
We extract its induced entity set $\mathcal{E}_i(v)$ from $\mathcal{C}_i^{top}$ by collecting the candidate entities at the variable endpoint of the refined triples.
When $v$ appears in multiple resolved constraints, we retain the intersection of their candidate sets when it is nonempty. Otherwise, we leave $v$ unbound and mark the associated constraints as unresolved, triggering textual recovery to avoid retrieval hallucination from inconsistent bindings.
The resulting binding set is denoted by $\mathcal{E}_b(v)$.
The consolidated binding set $\mathcal{E}_b(v)$ provides candidate assignments for subsequent constraint grounding.
For each resolved constraint, we retain triples that instantiate $v$ with an entity in $\mathcal{E}_b(v)$, yielding the updated structural pool $\hat{\mathcal{C}}_i^{top}$.
After that, we collect the associated source passages as textual evidence for each resolved atomic constraint.
For each resolved constraint $\tau_i\in\mathcal{P}_R$, the evidence $\mathcal{T}_i^{txt}$ is obtained by collecting the source passages associated with the triples in the updated structural pool $\hat{\mathcal{C}}_i^{top}$.

We then propagate bindings to two-variable constraints. For a constraint $\tau_j=(v_a,r_j,v_b)\in\mathcal{P}_2$, once a binding set $\mathcal{E}_b(v_a)$ is available, we instantiate $v_a$ with each candidate entity $e\in\mathcal{E}_b(v_a)$, yielding a one-variable constraint $(e,r_j,v_b)$. For each instantiation, we apply the same one-variable retrieval procedure to obtain candidate triples and their cross-encoder scores. We pool and deduplicate these triples across all instantiations, rank them by their scores, and retain the top-$K_s$ candidates as $\mathcal{C}_j^{top}$ for the original constraint $\tau_j$. This combined candidate pool undergoes the same sufficiency check to determine whether bindings for $v_b$ can be propagated. We repeat this process as new bindings become available.
After no further bindings can be propagated, any two-variable constraint with neither endpoint bound follows the textual recovery procedure above, yielding $\mathcal{T}_j^{txt}$ without inducing a binding.

\makeatletter
\let\csragVthreeTableType\ftype@table
\let\ftype@table\ftype@figure
\makeatother
\definecolor{AvgColor}{HTML}{7B3294}
\definecolor{StdColor}{HTML}{008837}

\providecommand{\metric}[1]{\textsc{#1}}
\newcommand{\avgstd}[2]{#1\,/\,#2}
\newcommand{\runstd}[2]{#1}
\newcommand{\bestavg}[1]{\textcolor{AvgColor}{\textbf{#1}}}
\newcommand{\secondavg}[1]{\textcolor{AvgColor}{\underline{#1}}}
\newcommand{\beststd}[1]{\textcolor{StdColor}{\textbf{#1}}}
\newcommand{\secondstd}[1]{\textcolor{StdColor}{\underline{#1}}}

\newsavebox{\mainbuildertablebox}
\newcommand{\mainbuildertablescale}{1}

\begin{table*}[!t]
\centering
\scriptsize
\renewcommand{\csragrowstretch}{0.86}
\setlength{\tabcolsep}{4.2pt}
\caption{Robustness results across KG builders on three datasets. 
For each method, the builder rows report the mean online evaluation results across five runs under the corresponding KG builder. All baselines are evaluated by us by following their released code. The last row of each method block reports $\mathrm{Avg}_{B}\uparrow$ and $\mathrm{Std}_{B}\downarrow$ across the five builder-level means for each metric column. Best results are \textbf{bolded} and second-best results are \underline{underlined}.}
\label{tab:main_builder}
\sbox{\mainbuildertablebox}{%
\begin{tabular}{l ccc ccc ccc}
\toprule
\multirow{2}{*}{\textbf{Builder}} 
& \multicolumn{3}{c}{\textbf{MuSiQue}} 
& \multicolumn{3}{c}{\textbf{2Wiki}} 
& \multicolumn{3}{c}{\textbf{HotpotQA}} \\
\cmidrule(lr){2-4} \cmidrule(lr){5-7} \cmidrule(lr){8-10}
& \metric{EM} & \metric{F1} & \metric{R@5} 
& \metric{EM} & \metric{F1} & \metric{R@5} 
& \metric{EM} & \metric{F1} & \metric{R@5} \\
\midrule

\rowcolor{gray!15} \multicolumn{10}{c}{\textbf{LightRAG}} \\
Qwen-3-1.7B   & \runstd{0.0}{0.0} & \runstd{2.4}{1.3} & \runstd{12.4}{1.4} & \runstd{2.9}{0.8} & \runstd{10.4}{1.4} & \runstd{37.5}{1.2} & \runstd{0.0}{0.0} & \runstd{4.8}{1.0} & \runstd{16.8}{1.3} \\
Qwen-3-14B    & \runstd{2.1}{1.0} & \runstd{5.9}{0.9} & \runstd{22.8}{1.2} & \runstd{5.5}{1.2} & \runstd{13.8}{1.7} & \runstd{46.2}{1.5} & \runstd{9.7}{1.0} & \runstd{14.8}{1.1} & \runstd{34.0}{1.3} \\
Llama-3.1-70B & \runstd{2.5}{0.8} & \runstd{7.4}{0.7} & \runstd{29.1}{0.8} & \runstd{7.4}{0.4} & \runstd{16.2}{0.5} & \runstd{53.5}{0.7} & \runstd{12.8}{1.3} & \runstd{19.3}{1.6} & \runstd{45.4}{1.3} \\
GPT-4o-mini   & \runstd{3.1}{0.9} & \runstd{9.7}{1.0} & \runstd{34.0}{0.9} & \runstd{9.8}{1.2} & \runstd{18.4}{1.3} & \runstd{59.8}{1.3} & \runstd{15.2}{1.2} & \runstd{26.9}{1.5} & \runstd{54.0}{1.1} \\
gemini-2.5-flash-lite & \runstd{2.9}{1.2} & \runstd{8.6}{1.4} & \runstd{32.6}{1.4} & \runstd{9.1}{1.3} & \runstd{17.3}{1.9} & \runstd{56.4}{1.5} & \runstd{15.8}{0.9} & \runstd{28.2}{1.1} & \runstd{56.3}{0.7} \\
\rowcolor{gray!6}
\textcolor{AvgColor}{Avg$_B$ (\% $\uparrow$)} / \textcolor{StdColor}{Std$_B$ ($\downarrow$)}
& \avgstd{2.1}{\beststd{1.1}}
& \avgstd{6.8}{\secondstd{2.5}}
& \avgstd{26.2}{\secondstd{7.9}}
& \avgstd{6.9}{\secondstd{2.5}}
& \avgstd{15.2}{\secondstd{2.9}}
& \avgstd{50.7}{8.0}
& \avgstd{10.7}{\secondstd{5.8}}
& \avgstd{18.8}{\secondstd{8.6}}
& \avgstd{41.3}{14.5} \\
\midrule

\rowcolor{gray!15} \multicolumn{10}{c}{\textbf{HippoRAG2}} \\
Qwen-3-1.7B   & \runstd{10.8}{0.9} & \runstd{17.6}{1.4} & \runstd{29.2}{1.8} & \runstd{38.6}{1.2} & \runstd{47.8}{1.3} & \runstd{60.7}{1.2} & \runstd{12.9}{1.3} & \runstd{22.8}{1.4} & \runstd{46.0}{1.5} \\
Qwen-3-14B    & \runstd{21.5}{0.5} & \runstd{31.5}{0.9} & \runstd{43.8}{1.1} & \runstd{44.8}{1.4} & \runstd{54.0}{1.5} & \runstd{67.5}{1.4} & \runstd{31.2}{0.6} & \runstd{46.5}{0.8} & \runstd{69.2}{0.9} \\
Llama-3.1-70B & \runstd{26.8}{1.1} & \runstd{37.8}{1.3} & \runstd{50.9}{1.2} & \runstd{48.3}{1.1} & \runstd{57.6}{1.3} & \runstd{71.2}{1.0} & \runstd{42.8}{0.8} & \runstd{58.6}{0.6} & \runstd{79.7}{0.4} \\
GPT-4o-mini   & \runstd{30.6}{1.0} & \runstd{42.2}{1.5} & \runstd{56.0}{1.8} & \runstd{51.5}{1.0} & \runstd{60.3}{1.0} & \runstd{74.9}{0.9} & \runstd{50.7}{1.2} & \runstd{64.7}{1.2} & \runstd{82.8}{0.9} \\
gemini-2.5-flash-lite & \runstd{32.1}{1.1} & \runstd{43.1}{1.4} & \runstd{57.4}{1.7} & \runstd{50.0}{1.3} & \runstd{58.6}{1.8} & \runstd{72.8}{2.0} & \runstd{52.3}{1.1} & \runstd{66.2}{1.3} & \runstd{84.0}{0.8} \\
\rowcolor{gray!6}
\textcolor{AvgColor}{Avg$_B$ (\% $\uparrow$)} / \textcolor{StdColor}{Std$_B$ ($\downarrow$)}
& \avgstd{24.4}{7.7}
& \avgstd{34.4}{9.4}
& \avgstd{47.5}{10.3}
& \avgstd{46.6}{4.6}
& \avgstd{55.7}{4.4}
& \avgstd{69.4}{\secondstd{5.0}}
& \avgstd{38.0}{14.6}
& \avgstd{51.8}{16.1}
& \avgstd{72.3}{14.2} \\
\midrule

\rowcolor{gray!15} \multicolumn{10}{c}{\textbf{GFM-RAG}} \\
Qwen-3-1.7B   & \runstd{11.2}{1.3} & \runstd{17.4}{1.1} & \runstd{33.5}{1.1} & \runstd{50.0}{0.8} & \runstd{60.2}{1.0} & \runstd{78.0}{1.7} & \runstd{15.1}{1.4} & \runstd{28.5}{1.2} & \runstd{60.2}{1.7} \\
Qwen-3-14B    & \runstd{21.4}{0.9} & \runstd{29.8}{1.4} & \runstd{45.6}{1.0} & \runstd{57.5}{1.2} & \runstd{66.8}{1.3} & \runstd{84.8}{1.7} & \runstd{35.4}{1.1} & \runstd{52.0}{0.7} & \runstd{78.6}{1.0} \\
Llama-3.1-70B & \runstd{26.6}{1.2} & \runstd{35.4}{1.7} & \runstd{52.1}{1.9} & \runstd{61.5}{1.7} & \runstd{70.5}{1.9} & \runstd{89.0}{1.6} & \runstd{45.9}{2.3} & \runstd{62.3}{2.2} & \runstd{85.4}{2.1} \\
GPT-4o-mini   & \runstd{30.5}{1.3} & \runstd{40.9}{1.5} & \runstd{57.6}{1.2} & \runstd{67.8}{0.8} & \runstd{76.9}{0.8} & \runstd{96.0}{0.9} & \runstd{51.2}{1.3} & \runstd{68.0}{1.2} & \runstd{87.8}{1.0} \\
gemini-2.5-flash-lite & \runstd{29.4}{1.4} & \runstd{38.7}{1.2} & \runstd{55.6}{1.6} & \runstd{65.3}{1.1} & \runstd{74.0}{1.6} & \runstd{91.8}{1.4} & \runstd{52.6}{1.8} & \runstd{70.4}{2.2} & \runstd{89.2}{2.2} \\
\rowcolor{gray!6}
\textcolor{AvgColor}{Avg$_B$ (\% $\uparrow$)} / \textcolor{StdColor}{Std$_B$ ($\downarrow$)}
& \avgstd{23.8}{7.1}
& \avgstd{32.4}{8.4}
& \avgstd{48.9}{8.7}
& \avgstd{\secondavg{60.4}}{6.3}
& \avgstd{\secondavg{69.7}}{5.8}
& \avgstd{\secondavg{87.9}}{6.2}
& \avgstd{40.0}{13.9}
& \avgstd{\secondavg{56.2}}{15.3}
& \avgstd{80.2}{10.7} \\
\midrule

\rowcolor{gray!15} \multicolumn{10}{c}{\textbf{NeuroPath}} \\
Qwen-3-1.7B   & \runstd{15.7}{1.8} & \runstd{22.8}{2.2} & \runstd{39.9}{2.4} & \runstd{50.8}{1.8} & \runstd{60.0}{1.7} & \runstd{76.8}{2.0} & \runstd{21.7}{2.1} & \runstd{34.2}{1.9} & \runstd{67.5}{1.9} \\
Qwen-3-14B    & \runstd{24.7}{1.4} & \runstd{34.6}{1.7} & \runstd{49.2}{1.8} & \runstd{57.4}{1.7} & \runstd{66.6}{2.0} & \runstd{83.1}{1.5} & \runstd{37.2}{1.6} & \runstd{50.8}{1.9} & \runstd{79.4}{1.5} \\
Llama-3.1-70B & \runstd{28.8}{1.3} & \runstd{39.7}{1.3} & \runstd{55.7}{1.7} & \runstd{60.2}{1.2} & \runstd{69.7}{0.9} & \runstd{87.0}{1.3} & \runstd{46.8}{1.5} & \runstd{60.9}{1.9} & \runstd{86.5}{1.9} \\
GPT-4o-mini   & \runstd{31.2}{1.5} & \runstd{43.5}{1.2} & \runstd{62.3}{1.4} & \runstd{63.4}{1.9} & \runstd{73.2}{1.5} & \runstd{91.8}{2.0} & \runstd{51.5}{2.0} & \runstd{64.9}{2.2} & \runstd{89.9}{2.2} \\
gemini-2.5-flash-lite & \runstd{32.5}{1.0} & \runstd{44.4}{1.5} & \runstd{63.6}{1.8} & \runstd{61.8}{1.7} & \runstd{71.0}{1.7} & \runstd{89.0}{1.5} & \runstd{52.2}{1.7} & \runstd{66.2}{1.6} & \runstd{91.1}{1.8} \\
\rowcolor{gray!6}
\textcolor{AvgColor}{Avg$_B$ (\% $\uparrow$)} / \textcolor{StdColor}{Std$_B$ ($\downarrow$)}
& \avgstd{\secondavg{26.6}}{6.1}
& \avgstd{\secondavg{37.0}}{7.9}
& \avgstd{\secondavg{54.1}}{8.8}
& \avgstd{58.7}{4.4}
& \avgstd{68.1}{4.6}
& \avgstd{85.5}{5.2}
& \avgstd{\secondavg{41.9}}{11.4}
& \avgstd{55.4}{11.9}
& \avgstd{\secondavg{82.9}}{\secondstd{8.7}} \\
\midrule

\rowcolor{gray!15} \multicolumn{10}{c}{\textbf{CS-RAG}} \\
Qwen-3-1.7B   & \runstd{29.9}{1.0} & \runstd{41.9}{1.3} & \runstd{61.7}{1.0} & \runstd{60.6}{1.6} & \runstd{68.8}{1.3} & \runstd{87.5}{1.2} & \runstd{51.6}{0.8} & \runstd{66.4}{0.8} & \runstd{82.2}{0.6} \\
Qwen-3-14B    & \runstd{31.0}{1.3} & \runstd{43.7}{1.7} & \runstd{63.7}{1.3} & \runstd{62.0}{1.1} & \runstd{70.0}{1.6} & \runstd{88.4}{1.4} & \runstd{52.7}{1.2} & \runstd{68.0}{0.9} & \runstd{84.1}{1.5} \\
Llama-3.1-70B & \runstd{31.8}{0.4} & \runstd{44.4}{0.5} & \runstd{64.5}{0.5} & \runstd{62.6}{1.3} & \runstd{70.6}{1.3} & \runstd{88.9}{1.3} & \runstd{53.4}{1.1} & \runstd{68.9}{0.9} & \runstd{85.0}{1.0} \\
GPT-4o-mini   & \runstd{32.2}{0.7} & \runstd{45.1}{0.9} & \runstd{65.0}{1.0} & \runstd{63.0}{0.6} & \runstd{71.0}{0.9} & \runstd{89.1}{1.0} & \runstd{54.3}{0.6} & \runstd{69.8}{1.0} & \runstd{86.9}{1.0} \\
gemini-2.5-flash-lite & \runstd{33.6}{1.6} & \runstd{45.6}{1.8} & \runstd{65.4}{1.5} & \runstd{63.8}{1.7} & \runstd{71.6}{1.6} & \runstd{89.6}{1.9} & \runstd{54.8}{0.8} & \runstd{70.6}{0.9} & \runstd{87.4}{1.6} \\
\rowcolor{gray!10}
\textbf{\textcolor{AvgColor}{Avg$_B$ (\% $\uparrow$)} / \textcolor{StdColor}{Std$_B$ ($\downarrow$)}}
& \avgstd{\bestavg{31.7}}{\secondstd{1.2}}
& \avgstd{\bestavg{44.1}}{\beststd{1.3}}
& \avgstd{\bestavg{64.1}}{\beststd{1.3}}
& \avgstd{\bestavg{62.4}}{\beststd{1.1}}
& \avgstd{\bestavg{70.4}}{\beststd{1.0}}
& \avgstd{\bestavg{88.7}}{\beststd{0.7}}
& \avgstd{\bestavg{53.4}}{\beststd{1.1}}
& \avgstd{\bestavg{68.7}}{\beststd{1.5}}
& \avgstd{\bestavg{85.1}}{\beststd{1.9}} \\
\bottomrule
\end{tabular}
}
\pgfmathparse{0.98 * \textwidth / \wd\mainbuildertablebox}
\xdef\mainbuildertablescale{\pgfmathresult}
\resizebox{0.98\textwidth}{!}{\usebox{\mainbuildertablebox}}
\end{table*}

\begin{figure*}[!t]
    \centering
    \includegraphics[width=0.95\linewidth]{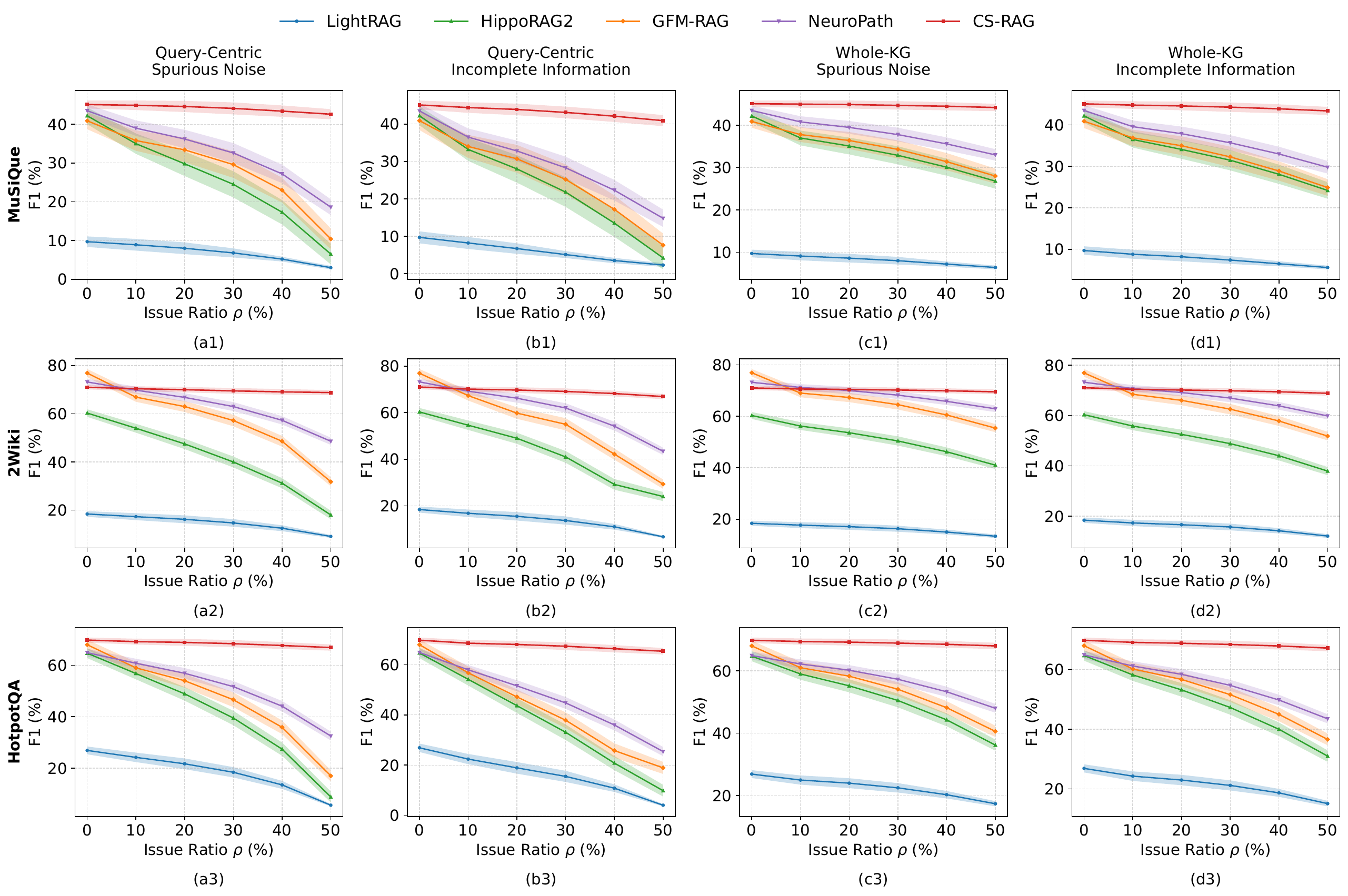} 
    \caption{F1 trends under controlled KG issue injection. Columns (a)--(d) correspond to Query-Centric Spurious Noise, Query-Centric Incomplete Information, Whole-KG Spurious Noise, and Whole-KG Incomplete Information, respectively. Rows 1--3 correspond to MuSiQue, 2Wiki, and HotpotQA, respectively. All baselines are evaluated by us by following their
released code. Lines show the mean F1, and shaded bands indicate the standard deviation over five runs.}
    \label{fig:robustness}
\end{figure*}

\subsection{LLM Generation}
Constraints in $\mathcal{P}_{\mathrm{obs}}$ introduce no binding target and typically restate query conditions. For each such $\tau_i$, we use the textual retrieval and reranking procedure described above, querying $U(q)$ with its linearized head--relation--tail form and retaining the top-$K_t$ passages as $\mathcal{T}_i^{txt}$ for auxiliary evidence.
Finally, we pack evidence from both subsets into constraint-aligned blocks:
\begin{equation}
B_i=\big[\,\mathrm{Lin}(\tau_i),\ \mathcal{T}_i^{txt}\big],\qquad
\mathcal{B}_q=\{B_i\}_{i=1}^{k},
\label{eq:evidence_blocks}
\end{equation}
yielding an ordered evidence context aligned with atomic constraints.
The resulting context $\mathcal{B}_q$ is fed into the language model to produce the final answer for the multi-hop query.

\section{Experiments}

\noindent\textbf{Datasets and Baselines.}
We evaluate CS-RAG on three standard multi-hop QA benchmarks: 2WikiMultihopQA~\cite{2wiki}, HotpotQA~\cite{hotpotqa}, and MuSiQue~\cite{Musique}.
Following HippoRAG~\cite{hipporag}, we evaluate on the same fixed subset of 1{,}000 validation queries per dataset, using identical queries across all methods and KG builders.
We compare CS-RAG with representative GraphRAG baselines, including LightRAG~\cite{lightrag}, HippoRAG2~\cite{hipporag2}, GFM-RAG~\cite{gfm-rag}, and NeuroPath~\cite{neuropath}. We evaluate all baselines using their released implementations and follow their official evaluation settings.
These baselines represent diverse ways of exploiting KG structure for downstream reasoning.

\noindent\textbf{Experimental Details.}
We use all-MiniLM-L6-v2~\cite{sentence-bert} as the dense retriever and bge-reranker-v2-m3~\cite{bge-reranker-v2-m3} as the cross-encoder reranker. GPT-4o-mini~\cite{gpt-4o} is used for atomic constraint planning and final answer generation. All experiments were run with Python 3.10 on an NVIDIA GeForce RTX 4090 server, with LLM-based KG construction, planning, and answer generation performed through API calls.
We evaluate QA performance using Exact Match (EM), token-level F1, and retrieval quality using Recall@5 over gold evidence.

\subsection{Robustness Across KG Builders}
\label{sec:main_builder}

To assess robustness to builder-induced KG variations, we evaluate CS-RAG on KGs constructed by different builders, including Qwen-3-1.7B, Qwen-3-14B~\cite{qwen3_technical_report_2025}, Llama-3.1-70B~\cite{llama3.1}, GPT-4o-mini~\cite{gpt-4o}, and Gemini-2.5-flash-lite~\cite{gemini2.5}. Following the setting used in HippoRAG2~\cite{hipporag2}, we keep the construction pipeline and tuple-extraction prompt fixed, varying only the LLM builder.
Table~\ref{tab:main_builder} reports Exact Match (EM) and token-level F1 for downstream QA, and Recall@5 for retrieval quality. The last row of each method block reports the average and standard deviation across the five KG builders, computed from the builder-level means, to summarize overall effectiveness and cross-builder sensitivity, respectively. Appendix~\ref{app:experiment_details} provides the further details.

\begingroup\brokenpenalty=5000
CS-RAG achieves the highest cross-builder averages across all datasets and metrics while maintaining consistently low standard deviations, indicating limited sensitivity to builder-induced KG variations. LightRAG sometimes has a relatively small standard deviation, but its much lower averages suggest uniformly weak performance. HippoRAG2 and GFM-RAG generally achieve higher cross-builder averages than LightRAG, but their graph propagation or graph-aware scoring still depends on reliable KG connectivity to guide evidence selection across successive retrieval hops. NeuroPath is relatively more stable due to path-level semantic control during graph traversal, but its lack of an explicit sufficiency check for structural continuation makes it less robust than CS-RAG. Builder sensitivity is weaker on 2Wiki than on MuSiQue and HotpotQA, likely because its more regular evidence structure and more explicit entity-bridging patterns limit error accumulation across repeated intermediate grounding steps. MuSiQue and HotpotQA require more cross-hop evidence chaining, allowing builder-induced issues to accumulate during retrieval.
\par\endgroup

\subsection{Robustness Under Controlled KG Issues}
\label{sec:robustness}

To evaluate the robustness of CS-RAG to spurious noise and incomplete information, we conduct experiments under controlled KG quality degradation. We consider two injection scopes. Query-Centric Issue Injection targets local 2-hop entity neighborhoods to affect reasoning paths, while Whole-KG Issue Injection distributes issues across the full query-specific KG. Appendix~\ref{app:issue_injection} details both issue modes. We vary the injection ratio $\rho$ over $\{0.1,0.2,0.3,0.4,0.5\}$ and repeat each setting with five different seeds on GPT-4o-mini-constructed KGs, reporting results in Figure~\ref{fig:robustness}.

\begingroup\brokenpenalty=5000\looseness=1
Under Query-Centric Issue Injection (Figure~\ref{fig:robustness}(a)--(b)), CS-RAG maintains nearly stable F1 across all three datasets as $\rho$ increases, while the baselines degrade.  Under spurious noise, corrupted but plausible triples can participate in graph matching and propagation, causing baselines to drift from the intended evidence chain. NeuroPath is relatively more stable due to path-level selection, but its performance still drops when noisy local candidates enter the reasoning path. Under incomplete information, baselines suffer from missing local bridges because they tend to continue structural retrieval even when the required evidence chain is under-supported. CS-RAG is much less affected because sufficiency checking and textual recovery reduce forced structural continuation. Under Whole-KG Issue Injection (Figure~\ref{fig:robustness}(c)--(d)), baseline degradation is generally more gradual because fewer injected issues directly intersect reasoning paths. F1 still declines as globally distributed issues affect graph matching and propagation, particularly for baselines dependent on reliable connectivity. CS-RAG remains nearly stable across issue ratios, indicating that its combined constraint-based retrieval, propagation control, and textual recovery tolerate broadly distributed KG issues.
\par\endgroup

\subsection{Ablation Study}
\label{sec:ablation}

\begin{table}[!t]
    \centering
    \scriptsize
    \setlength{\tabcolsep}{1.6pt}
    \caption{Ablation results of EM (\%) and F1 (\%) under different builders on HotpotQA.}
    \label{tab:ablation_main}
\renewcommand{\csragrowstretch}{0.98}
    \scalebox{\mainbuildertablescale}{%
    \begin{tabular}{lcccccc}
        \toprule
        \multirow{2}{*}{Method} &
        \multicolumn{2}{c}{\textbf{GPT-4o-mini}} &
        \multicolumn{2}{c}{\textbf{Llama-3.1-70B}} &
        \multicolumn{2}{c}{\textbf{Qwen-3-14B}} \\
        \cmidrule(lr){2-3} \cmidrule(lr){4-5} \cmidrule(lr){6-7}
        & EM & F1 & EM & F1 & EM & F1 \\
        \midrule
        CS-RAG & 54.3 & 69.8 & 53.4 & 68.9 & 52.7 & 68.0 \\
       \rowcolor[gray]{0.93}
\multicolumn{7}{l}{\textit{\strut Constraint-based Retrieval}} \\
        w/o Planner & 49.8 & 64.0 & 47.1 & 61.3 & 45.6 & 59.3 \\
        w/o Relation & 51.5 & 66.4 & 50.4 & 64.8 & 49.3 & 63.8 \\
        w/o Constraint & 51.0 & 65.8 & 50.0 & 64.3 & 48.6 & 62.7 \\
        \rowcolor[gray]{0.93}
\multicolumn{7}{l}{\textit{\strut Sufficiency-Guided Propagation and Recovery}} \\
        w/o Check & 50.3 & 64.7 & 47.8 & 62.0 & 46.5 & 60.2 \\
        w/o Dist & 52.1 & 66.7 & 50.5 & 65.2 & 49.5 & 63.5 \\
        w/ Ent. & 53.4 & 67.1 & 51.8 & 66.0 & 50.4 & 65.6 \\
        \bottomrule
    \end{tabular}
    }
\end{table}
We conduct ablations of constraint-based retrieval and sufficiency checking on 1,000 HotpotQA queries, changing one design choice at a time. For constraint-based retrieval, w/o Planner removes atomic constraint planning and uses the original query as a single semantic retrieval request. w/o Constraint keeps the LLM-generated atomic constraints but replaces fine-grained constraint retrieval with coarse semantic matching. w/o Relation removes relation alignment and reranks candidates using only entity information. For Sufficiency-Guided Propagation and Recovery,
w/o Check disables textual recovery and always propagates available structural candidates to subsequent hops.
w/o Dist replaces the distribution-based criterion with a fixed binary score threshold,
classifying a constraint as resolved only when the highest candidate score after sigmoid transformation exceeds 0.5.
In w/ Ent., we replace $N_{\mathrm{eff}}$ with the entropy of the same candidate distribution, normalized by the maximum possible entropy for its pool. We select the best threshold from $\{0.2, 0.4, 0.6, 0.8\}$ on the same 100 development queries and keep it fixed for evaluation on the 1,000 queries. A constraint is classified as resolved when its normalized entropy is at most the selected threshold.

The HotpotQA results in Table~\ref{tab:ablation_main} show that w/o Planner suffers the largest drop, highlighting the importance of explicit hop-wise targets and intermediate dependencies. The drop in w/o Constraint despite retaining planning further indicates that coarse semantic matching cannot enforce fine-grained constraint execution and may admit target-misaligned evidence. The consistent drop in w/o Relation across builders shows that relation matching helps distinguish entity-related candidates from constraint-compatible evidence. The drop in w/o Check supports the joint contribution of propagation control and textual recovery. With recovery retained, the lower performance of w/o Dist supports the distribution-based criterion over the fixed score threshold used here. The entropy variant also underperforms across builders. Entropy summarizes uncertainty over the full candidate distribution, while $N_{\mathrm{eff}}$ responds more directly to concentrated probability mass among leading candidates. Accepted bindings serve as anchors for later hops, making the sufficiency decision consequential for downstream reasoning. The consistent gap supports using $N_{\mathrm{eff}}$ as an empirical signal for controlling structural propagation. Further analysis in Appendix~\ref{app:sufficiency_correctness} (Table~\ref{tab:builder_sufficiency}) evaluates the sufficiency check across KG builders. The resolved-hop rate varies across builders, while a high proportion of resolved hops retain at least one evidence-supported binding, supporting the use of the sufficiency check to guide structural propagation under varying KG conditions. Additional ablation results for other modules are provided in Appendix~\ref{appsec:ablation}.

\begin{table}[!htbp]
\renewcommand{\csragrowstretch}{0.98}
\centering
\scriptsize
\setlength{\tabcolsep}{2pt}
\caption{Effect of the representative $\gamma$-values from the searched range across datasets. Res. denotes the percentage of resolved hops. Time measures online retrieval latency, Tokens counts retrieved evidence tokens passed to final QA. \textbf{Bold} marks the trade-off setting used in the main experiments.}
\label{tab:gamma_sensitivity}
\pgfmathsetlengthmacro{\gammasensitivitywidth}{\columnwidth / \mainbuildertablescale}
\scalebox{\mainbuildertablescale}{%
\begin{tabular*}{\gammasensitivitywidth}{@{\extracolsep{\fill}}llccccc@{}}
\toprule
Dataset & $\gamma$ & Res. (\%) & EM (\%)  & F1 (\%)  & Time(s) & Tokens \\
\midrule
\multirow{5}{*}{MuSiQue}
& 1.2 & 25.6 & 34.1 & 47.9 & 8.8 & 5.4k\\
& \textbf{1.5} & \textbf{44.9} & \textbf{32.2} & \textbf{45.1} & \textbf{6.3} & \textbf{3.9k} \\
& 1.8 & 57.7 & 30.2 & 42.1 & 5.8 & 3.2k \\
& 2.1 & 69.8 & 27.4 & 39.8 & 4.2 & 2.6k \\
& 2.4 & 80.9 & 25.7 & 35.2 & 3.3 & 1.9k \\
\midrule
\multirow{5}{*}{2Wiki}
& 1.2 & 24.2 & 61.9 & 66.2 & 3.0 & 1.9k \\
& 1.4 & 41.0 & 62.8 & 69.1 & 2.3 & 1.4k \\
& \textbf{1.6} & \textbf{45.8} & \textbf{63.0} & \textbf{71.0} & \textbf{1.7} & \textbf{1.1k} \\
& 1.8 & 79.1 & 63.6 & 70.4 & 1.4 & 0.9k \\
& 2.0 & 86.7 & 60.8 & 68.5 & 1.3 & 0.8k \\
\midrule
\multirow{5}{*}{HotpotQA}
& 1.2 & 20.4 & 57.7 & 72.1 & 3.2 & 2.0k \\
& 1.6 & 44.1 & 55.6 & 70.1 & 2.6 & 1.5k \\
& \textbf{2.0} & \textbf{61.5} & \textbf{54.3} & \textbf{69.8} & \textbf{2.2} & \textbf{1.2k} \\
& 2.4 & 74.6 & 51.9 & 66.7 & 1.9 & 1.0k \\
& 2.8 & 83.4 & 47.6 & 62.4 & 1.7 & 0.8k \\
\bottomrule
\end{tabular*}%
}
\end{table}

\begin{table*}[!t]
\centering
\scriptsize
\setlength{\tabcolsep}{3.2pt}
\caption{Online efficiency on all evaluation queries under GPT-4o-mini-constructed KGs. Time (seconds) and Tokens are averaged over all evaluation queries. Tokens count input and output tokens from planning through final QA.}
\label{tab:efficiency_e2e_sampled}
\renewcommand{\csragrowstretch}{0.98}
\pgfmathsetlengthmacro{\supplementalefficiencywidth}{\textwidth / \mainbuildertablescale}
\scalebox{\mainbuildertablescale}{%
\begin{tabular*}{\supplementalefficiencywidth}{@{\extracolsep{\fill}}lcccccccccccc@{}}
\toprule
\multirow{2}{*}{Method}
& \multicolumn{4}{c}{MuSiQue}
& \multicolumn{4}{c}{2WikiMultihopQA}
& \multicolumn{4}{c}{HotpotQA} \\
\cmidrule(lr){2-5} \cmidrule(lr){6-9} \cmidrule(lr){10-13}
& Time (s) & Tokens & EM (\%) & F1 (\%)
& Time (s) & Tokens & EM (\%) & F1 (\%)
& Time (s) & Tokens & EM (\%) & F1 (\%) \\
\midrule
LightRAG & 5.3 & 7.7k & 3.1 & 9.7 & 4.6 & 4.8k & 9.8 & 18.4 & 4.9 & 4.9k & 15.2 & 26.9 \\
HippoRAG2 & 5.4 & 4.4k & 30.6 & 42.2 & 3.8 & 1.4k & 51.5 & 60.3 & 4.5 & 1.6k & 50.7 & 64.7 \\
GFM-RAG & 5.1 & 4.5k & 30.5 & 40.9 & 3.6 & 1.5k & 67.8 & 76.9 & 4.1 & 1.7k & 51.2 & 68.0 \\
NeuroPath & 8.9 & 5.7k & 31.2 & 43.5 & 7.1 & 2.7k & 63.4 & 73.2 & 8.0 & 2.9k & 51.5 & 64.9 \\
\midrule
CS-RAG & 8.4 & 5.1k & 32.2 & 45.1 & 3.5 & 2.1k & 63.0 & 71.0 & 4.1 & 2.3k & 54.3 & 69.8 \\
\bottomrule
\end{tabular*}%
}
\end{table*}

\subsection{Hyperparameter Study}
\label{appsec:rq7_hparam}

We analyze how the sufficiency threshold $\gamma$ shifts the balance between structural propagation and textual recovery across the three datasets (Table~\ref{tab:gamma_sensitivity}).
As $\gamma$ increases,
more constraints pass the sufficiency check for structural propagation, reducing retrieval latency and retrieved evidence tokens.
On MuSiQue, lower $\gamma$ improves EM/F1 despite the added retrieval cost, consistent with textual recovery limiting the propagation of uncertain bindings through longer dependency chains. HotpotQA follows a similar but milder accuracy-efficiency trade-off. For 2Wiki, QA is strongest at moderate settings of $\gamma$, consistent with less need for recovery on its shorter, more directly executable paths and greater exposure to weak bindings at higher thresholds. For each dataset, $\gamma$ should be selected to balance the cost of unnecessary recovery against the downstream risk of propagating uncertain bindings. Additional sensitivity analyses of $m$, $R$, $K_s$, and $K_t$ are provided in Appendix~\ref{appsec:hparam_supp}.

\subsection{Efficiency Analysis}
\label{appsec:efficiency}

We assess online efficiency on all evaluation queries with GPT-4o-mini for answer generation. We compare CS-RAG with the baselines on GPT-4o-mini-constructed KGs and measure its cost across KG builders. We measure online costs from planning through final QA, excluding the shared KG construction cost. The per-query LLM token budget is 5k for 2Wiki and HotpotQA, rising to 8k for MuSiQue, whose longer, more interdependent reasoning chains require more room to retain evidence. We report mean per-query latency and LLM input/output tokens alongside EM/F1.
On MuSiQue, longer constraint chains increase CS-RAG's reranking and recovery costs relative to HippoRAG2 and GFM-RAG (Table~\ref{tab:efficiency_e2e_sampled}). NeuroPath relies on LLMs to select and expand paths, while CS-RAG plans atomic constraints once and controls subsequent retrieval without repeated generation of intermediate reasoning. This reduces generative overhead. CS-RAG's higher token use than HippoRAG2 on 2Wiki and HotpotQA accompanies comparable or lower latency. Table~\ref{tab:csrag_builder_cost} shows a consistent increase in CS-RAG's cost as KG quality declines, with MuSiQue remaining the most expensive dataset under every builder. This pattern is consistent with more unresolved constraints triggering textual recovery on weaker KGs and with longer reasoning chains requiring more constraint-level checks on MuSiQue. The fixed recovery budget limits the cost of each unresolved constraint.

\begin{table}[!htbp]
\centering
\scriptsize
\setlength{\tabcolsep}{2pt}
\caption{Online cost of CS-RAG across KG builders on all evaluation queries. Time is measured in seconds per query. Tokens count input and output tokens from planning through final QA.}
\label{tab:csrag_builder_cost}
\renewcommand{\csragrowstretch}{0.98}
\pgfmathsetlengthmacro{\csragbuildercostwidth}{\columnwidth / \mainbuildertablescale}
\scalebox{\mainbuildertablescale}{%
\begin{tabular*}{\csragbuildercostwidth}{@{\extracolsep{\fill}}lcccccc@{}}
\toprule
\multirow{2}{*}{KG builder}
& \multicolumn{2}{c}{MuSiQue}
& \multicolumn{2}{c}{2Wiki}
& \multicolumn{2}{c}{HotpotQA} \\
\cmidrule(lr){2-3} \cmidrule(lr){4-5} \cmidrule(lr){6-7}
& Time & Tokens & Time & Tokens & Time & Tokens \\
\midrule
Qwen-3-1.7B   & 10.7 & 6.4k & 4.8 & 3.7k & 6.4 & 4.1k \\
Qwen-3-14B    & 10.0 & 5.9k & 4.0 & 2.8k & 5.3 & 3.1k \\
Llama-3.1-70B & 9.3 & 5.5k & 3.6 & 2.4k & 4.6 & 2.7k \\
GPT-4o-mini   & 8.4 & 5.1k & 3.5 & 2.1k & 4.1 & 2.3k \\
Gemini-2.5-flash-lite & 8.1 & 5.0k & 3.3 & 2.0k & 3.9 & 2.2k \\
\bottomrule
\end{tabular*}%
}
\end{table}

\subsection{Sensitivity to Planner Quality}
\label{sec:planner_error_stats}

 We assess CS-RAG's sensitivity to planner choice using five planners across three datasets, with GPT-4o-mini fixed for answer generation. Planning quality is measured by hop-error rate (HE) and plan-invalid rate (PI). We use GPT-4o-mini to assess each generated constraint against the query and its gold supporting evidence. A constraint is marked as incorrect if its relation meaning is drifted or misaligned, or if its variable is missing, swapped, or inconsistently reused. Hop-error rate is the fraction of incorrect constraints among all generated constraints. Plan-invalid rate is the fraction of queries whose generated plan contains at least one incorrect constraint. We also report EM and F1 to measure downstream QA performance. All four metrics are computed over the same 1,000 queries per dataset and reported as percentages. Table~\ref{tab:planner_error_stats_main} shows that planning errors remain limited across the five planners, with most queries receiving valid constraint plans even from the smallest model. MuSiQue consistently has the highest HE and PI, suggesting that query complexity influences planning across models. Although smaller planners exhibit higher error rates, downstream QA performance remains relatively stable. This stability stems from the simplicity of the structured planning schema, whose explicit relation and variable requirements are reinforced by representative few-shot examples in the prompt. Thus, CS-RAG shows limited dependence on planner choice.

 \begin{table}[!htbp]
\centering
\scriptsize
\setlength{\tabcolsep}{3pt}
\caption{Quality evaluation of different planners (\%). HE and PI denote hop-error and plan-invalid rates. EM and F1 measure downstream QA performance.}
\label{tab:planner_error_stats_main}
\renewcommand{\csragrowstretch}{0.98}
\pgfmathsetlengthmacro{\plannertablewidth}{\columnwidth / \mainbuildertablescale}
\scalebox{\mainbuildertablescale}{%
\begin{tabular*}{\plannertablewidth}{@{\extracolsep{\fill}}l cc cc cc@{}}
\toprule
Planner & \multicolumn{2}{c}{MuSiQue} & \multicolumn{2}{c}{2Wiki} & \multicolumn{2}{c}{HotpotQA} \\
\cmidrule(lr){2-3}\cmidrule(lr){4-5}\cmidrule(lr){6-7}
& HE & PI & HE & PI & HE & PI \\
\midrule
Qwen-3-1.7B   & 9.7 & 10.4 & 7.9 & 8.0 & 8.5 & 8.9 \\
Qwen-3-14B    & 7.7 & 9.6 & 6.1 & 7.0 & 6.6 & 7.7 \\
Llama-3.1-70B & 6.5 & 8.7 & 4.7 & 6.1 & 5.3 & 6.8 \\
GPT-4o-mini   & 5.2 & 8.0 & 3.6 & 5.4 & 4.1 & 6.1 \\
Gemini-2.5-flash-lite & 5.0 & 7.8 & 4.3 & 5.6 & 4.6 & 6.4 \\
\midrule
& EM & F1 & EM & F1 & EM & F1 \\
\midrule
Qwen-3-1.7B   & 28.5 & 40.7 & 59.8 & 67.7 & 50.0 & 64.3 \\
Qwen-3-14B    & 29.8 & 43.4 & 60.9 & 68.5 & 51.7 & 67.5 \\
Llama-3.1-70B & 31.1 & 44.3 & 61.8 & 69.2 & 52.9 & 68.4 \\
GPT-4o-mini   & 32.2 & 45.1 & 63.0 & 71.0 & 54.3 & 69.8 \\
Gemini-2.5-flash-lite & 32.9 & 45.4 & 63.4 & 71.3 & 52.6 & 68.2 \\
\bottomrule
\end{tabular*}%
}
\end{table}
\renewcommand{\csragrowstretch}{0.90}

\FloatBarrier
\makeatletter
\let\ftype@table\csragVthreeTableType
\makeatother
\section{Conclusion}
In this paper, we propose CS-RAG, a GraphRAG framework for robust multi-hop retrieval across LLM-constructed KGs of varying quality. Our cross-builder analysis identifies shared error patterns with different proportions, grouped into spurious noise causing retrieval drift and incomplete information causing retrieval hallucination. CS-RAG combines constraint-based retrieval with sufficiency-guided propagation and recovery to mitigate these two retrieval failures. Specifically, the former uses atomic constraints to restrict evidence selection at each hop, while the latter checks candidate concentration to control binding propagation and triggers textual recovery for unresolved constraints. Experiments on three multi-hop QA benchmarks show that CS-RAG is more robust across KG builders and under controlled KG issue injection.

\clearpage
\setlength{\textfloatsep}{\csragOriginalTextFloatSep}
\setlength{\dbltextfloatsep}{\csragOriginalDblTextFloatSep}
\setlength{\floatsep}{\csragOriginalFloatSep}
\setlength{\dblfloatsep}{\csragOriginalDblFloatSep}
\setlength{\intextsep}{\csragOriginalInTextSep}
\bibliography{mybib}
\bibliographystyle{ACM-Reference-Format}

\appendix %
\renewcommand{\csragrowstretch}{0.88}
\setlength{\textfloatsep}{10pt plus 2pt minus 2pt}
\setlength{\floatsep}{8pt plus 2pt minus 1pt}
\setlength{\intextsep}{6pt plus 1pt minus 1pt}
\captionsetup[table]{skip=4pt}

\section{Representative Incorrect Patterns}
\label{app:issue_patterns}

In this section, we provide a detailed description of six representative incorrect patterns observed in builder-extracted tuples. Specifically, we manually check whether the extracted KG captures the relation required by the supporting evidence at each reasoning hop. We report the proportions of correct extractions and error types over all evaluated hops, counting absent required bridging relations as missing bridge edges.

\noindent\textbf{Over-generalized Relation.}
This pattern refers to cases where the extracted tuple replaces a specific relation stated in the source evidence with a broader or less informative relation.
For example, the sentence states that \emph{Jeremy Renner was nominated for the Academy Award for Best Actor for his role in The Hurt Locker}, while the extracted tuple is effectively \emph{Jeremy Renner -- associated with -- Academy Award for Best Actor}. The extracted relation preserves a topical connection between the two entities, but it drops the precise nomination relation required by the evidence.

\noindent\textbf{Mis-bound Relation.}
This pattern refers to cases where the extracted tuple links entities that appear in the same sentence, but the relation between them is not actually expressed in the source evidence.
For example, the sentence states that \emph{Karl Marx was born in Trier and studied at the University of Bonn}, while the extracted tuple becomes \emph{Trier -- related\_to -- Bonn}. Although both entities are mentioned, the sentence does not express such a direct relation.

\noindent\textbf{Semantic Flip.}
This pattern refers to cases where the extracted tuple reverses or distorts the original semantic meaning of the source evidence.
For example, the sentence says that \emph{the report doesn't find that Drug A is associated with increased risk of stroke}, while the extracted tuple becomes \emph{Drug A -- associated with -- increased\_risk\_stroke}. The tuple remains topically related, but its meaning is opposite to that of the source sentence.

\noindent\textbf{Missing Bridge Edge.}
This pattern refers to cases where the extracted tuples fail to provide a necessary connection between two pieces of evidence, leaving the reasoning chain structurally incomplete.
For example, the graph contains \emph{iPhone 13 -- powered\_by -- A15 Bionic} and \emph{Apple Inc. -- headquartered\_in -- Cupertino}, but misses the bridge \emph{A15 Bionic -- designed\_by -- Apple Inc.}. As a result, the intended multi-hop connection cannot be completed.

\noindent\textbf{Dropped Qualifier.}
This pattern refers to cases where the extracted tuple preserves the main fact but omits a qualifier that is necessary for a complete or precise interpretation.
For example, the sentence states that \emph{In 2002, Elon Musk founded SpaceX}, while the extracted tuple keeps \emph{Elon Musk -- founded -- SpaceX} but drops the qualifier \emph{in 2002}. The extracted fact is only partially preserved.

\noindent\textbf{Others.}
This pattern refers to a case set of low-frequency incorrect cases that usually have limited influence on downstream retrieval.
For example, one case is incomplete entity surface extraction, where the extracted entity name is only partially preserved, such as extracting \emph{Musk} instead of \emph{Elon Musk}. Another case is off-chain redundant tuple extraction, where the builder extracts an additional tuple that is loosely related to the source sentence but falls outside the evidence chain required by the query.  For example, from a sentence used to support \emph{No Cross, No Crown -- written during imprisonment in -- Tower of London}, the builder may additionally extract \emph{Tower of London -- located\_in -- London}. This tuple is sentence-related, but is irrelevant to the reasoning chain.

\begingroup
\section[Analysis of $N_{\mathrm{eff}}$]{Analysis of $N_{\mathrm{eff}}$}
\label{app:neff}

This appendix analyzes the concentration properties of $N_{\mathrm{eff}}$ and the interpretation of its threshold.
Here $\{p_{i,c}\}_{c\in\mathcal{C}_i^{top}}$ is the normalized candidate distribution computed from cross-encoder scores, and $\mathcal{C}_i^{top}$ is the refined structural candidate pool for constraint $\tau_i$.

Let $p$ be any probability distribution on a finite candidate set $\mathcal{C}$ with $n=|\mathcal{C}|$, and define
\begin{equation}
N_{\mathrm{eff}}(p)\triangleq\frac{1}{\sum_{c\in\mathcal{C}}p_c^2}.
\label{eq:app_neff_general}
\end{equation}
Cauchy--Schwarz and $p_c\le1$ give
\[
\frac{1}{n}=\frac{(\sum_c p_c)^2}{n}
\le \sum_c p_c^2
\le \sum_c p_c=1.
\]
Thus $1\le N_{\mathrm{eff}}(p)\le n$, with $N_{\mathrm{eff}}(p)=1$ for a point mass and $N_{\mathrm{eff}}(p)=n$ for uniform $p$.
More generally, $N_{\mathrm{eff}}$ represents the equivalent number of equally weighted candidates, yielding $k$ when probability mass is equally distributed over $k$ candidates.

We next show that the sufficiency check $N_{\mathrm{eff}}(p)\le\gamma$ imposes a lower bound on the probability mass assigned to the top-ranked candidate.
Let $p_{\max}=\max_{c\in\mathcal{C}}p_c$.
For every $c$, $p_c\le p_{\max}$ implies $p_c^2\le p_{\max}p_c$.
Summing over $c$ yields
\begin{equation}
\frac{1}{N_{\mathrm{eff}}(p)}=\sum_c p_c^2
\le p_{\max}\sum_c p_c=p_{\max}.
\label{eq:app_pmax_lb}
\end{equation}

If the sufficiency check accepts $\tau_i$, i.e., $N_{\mathrm{eff}}(\tau_i)\le\gamma$, then
\begin{equation}
\max_{c\in\mathcal{C}_i^{top}}p_{i,c}
\ge\frac{1}{N_{\mathrm{eff}}(\tau_i)}
\ge\frac{1}{\gamma}.
\label{eq:app_accept_lb}
\end{equation}
The threshold thus bounds the effective number of candidates for an accepted constraint. A smaller $\gamma$ permits fewer equivalent alternatives and requires a higher minimum probability mass for the top-ranked candidate.
\par
\endgroup

\section{Details of Experiments.}
\label{app:experiment_details}
\noindent\textbf{Offline KG construction.}
We follow the offline KG construction procedure of HippoRAG2~\cite{hipporag2}. Given a query $q$, we obtain its candidate passage set $U(q)$ from the dataset. We then use an LLM-based KG builder to extract relational tuples $(h,r,t)$ from the passages in $U(q)$ and construct a query-specific KG. The resulting graphs are merged into a shared KG for retrieval across queries. Each extracted tuple is attached with provenance pointers to its source passage. In the cross-builder evaluation, we keep the construction pipeline and tuple-extraction prompt fixed.

\noindent\textbf{Retrieval metric.}
We compute Recall@5 by projecting retrieved items into a unified evidence list. Specifically, we aggregate retrieved evidence from all atomic constraints, deduplicate candidates by passage identifier, and rank them by a global score, where repeated passages take the maximum reranker score across constraints. 
Recall@5 is the mean fraction of gold evidence passages covered by the top-5 ranked passages across queries.

\noindent\textbf{Hyperparameter details.}
Table~\ref{tab:used_hyperparams} lists candidate ranges and settings for key hyperparameters in the main experiments. 
For each dataset, we randomly sample 100 queries for hyperparameter selection, with no overlap with the 1,000-query evaluation subset. We select hyperparameters using GPT-4o-mini-constructed KGs to balance QA performance and retrieval cost. The selected configuration is then fixed across all KG builders for evaluation.

\begin{table}[!htbp]
\centering
\small
\setlength{\tabcolsep}{2pt}
\caption{Hyperparameter settings used in the main experiments.}
\label{tab:used_hyperparams}
\renewcommand{\csragrowstretch}{0.94}
\begin{tabular}{@{}llccc@{}}
\toprule
\multirow{2}{*}{Parameter} 
& \multirow{2}{*}{Range} 
& \multicolumn{3}{c}{Used} \\
\cmidrule(lr){3-5}
& & MuSiQue & 2Wiki & HotpotQA \\
\midrule
\multicolumn{5}{l}{\textit{Constraint retrieval}} \\
Rel. variants \(m\) 
& \(\{1,2,3,4,5\}\) 
& \(4\) & \(3\) & \(3\) \\
Rel.-filtered cand. \(R\) 
& \(\{6,8,10,12,14\}\) 
& \(10\) & \(10\) & \(10\) \\
Struct. top-\(K_s\) 
& \(\{1,3,5,7,9\}\) 
& \(7\) & \(5\) & \(7\) \\
\midrule
\multicolumn{5}{l}{\textit{Sufficiency-guided propagation and recovery}} \\
Text top-\(K_t\) 
& \(\{1,3,5,7,9\}\) 
& \(7\) & \(5\) & \(7\) \\
Threshold \(\gamma\) 
& \([1.0,3.0]\), step \(0.1\) 
& \(1.5\) & \(1.6\) & \(2.0\) \\
\bottomrule
\end{tabular}
\end{table}

\section{Details of Controlled KG Issue Injection}
\label{app:issue_injection}

\begin{table}[!htbp]
\centering
\small
\setlength{\tabcolsep}{4pt}
\caption{Representative examples of controlled KG issue injection. Unchanged components are written once, with changes indicated by $\rightarrow$.}
\label{tab:issue_injection_examples}
\renewcommand{\csragrowstretch}{0.95}
\begin{tabular}{@{}p{0.20\linewidth}p{\dimexpr0.80\linewidth-8pt\relax}@{}}
\toprule
Issue pattern & Original $\rightarrow$ injected form \\
\midrule
Over-generalized Relation & Apollo 13 -- (nominated for $\rightarrow$ associated with) -- Academy Award for Best Picture \\
Mis-bound Relation & Marie Curie -- born in -- (Warsaw $\rightarrow$ Paris) \\
Semantic Flip & Mary Shelley -- (mother of $\rightarrow$ child of) -- Percy Florence Shelley \\
Missing Bridge Edge & The Old Man and the Sea -- written by -- Ernest Hemingway $\rightarrow$ triple removed \\
Dropped Qualifier & Kathryn Bigelow -- won -- (2010 Academy Award for Best Director $\rightarrow$ Academy Award for Best Director) \\
\bottomrule
\end{tabular}
\end{table}

All controlled injections in Section~\ref{sec:robustness} are applied to the GPT-4o-mini constructed KG under two scopes. For query-centric injection, we match the observed constant entities in the atomic constraint plan to KG nodes and restrict injection to their 2-hop neighborhood. For whole-KG injection, all triples in the full query-specific KG are eligible. Within each scope, we consider spurious noise and incomplete information, instantiated by the representative incorrect patterns in Appendix~\ref{app:issue_patterns}. For each triple, injection is triggered with probability $\rho \in \{0.1,0.2,0.3,0.4,0.5\}$ by uniformly sampling an applicable pattern from the corresponding mode, leaving the triple unchanged if none applies. The construction rules are detailed below, with representative examples in Table~\ref{tab:issue_injection_examples}.

\noindent\textbf{1. Spurious noise.}

\textbf{Over-generalized Relation.}
For a selected triple $(h,r,t)$, we keep $h$ and $t$ unchanged and replace $r$ with a generic relation sampled from \{``related to'', ``associated with'', ``connected to'', ``linked with''\}.

\textbf{Mis-bound Relation.}
 For a selected triple $(h,r,t)$, we keep the relation $r$ unchanged and replace one endpoint. The replacement entity is sampled from the entity pool of the current injection scope. If entity type information is available, the replacement is restricted to entities with the same type as the original endpoint. The replacement excludes both original endpoints.

\textbf{Semantic Flip.}
        For a selected triple $(h,r,t)$, we keep the head $h$ and tail $t$ unchanged, and replace the relation $r$ with an opposite or semantically incompatible relation. 
        The predefined phrase-level mappings include
        born $\leftrightarrow$ died,
        start $\leftrightarrow$ end,
        before $\leftrightarrow$ after,
        parent $\leftrightarrow$ child,
        father $\leftrightarrow$ child,
        mother $\leftrightarrow$ child,
        spouse $\leftrightarrow$ sibling,
        winner $\leftrightarrow$ loser,
        win $\leftrightarrow$ lose,
        contain $\leftrightarrow$ located in,
        employer $\leftrightarrow$ employee,
        member $\leftrightarrow$ opponent, and
        director $\leftrightarrow$ actor.
        If a relation contains one phrase in the mapping set, the first matched phrase is replaced by its paired phrase. 

\noindent\textbf{2. Incomplete information.}
\par\nopagebreak[4]

\textbf{Missing Bridge Edge.}
        We remove selected triples from the target scope. 
 In the query-centric setting, a triple is prioritized as a bridge candidate if it meets either condition. First, its two endpoints are located at different depths in the local 2-hop neighborhood. Second, its relation matches one of the relation variants generated during atomic constraint planning, where exact match and containment-based fuzzy match are both allowed. If no prioritized bridge candidate is found in the target scope, the deletion is sampled from the eligible triples there.

\textbf{Dropped Qualifier.}
        For a selected KG triple, we remove qualifier fields or qualifying information embedded in the triple, while preserving the main relation and keeping the associated evidence text and source passages intact.
        The removable fields are \{qualifier, qualifiers, time, date, year, start\_time, end\_time, point\_in\_time, location\}.

\section{Supplementary Results}
\label{Supplementary Results}

\subsection{Supplementary of Ablation Study}
\label{appsec:ablation}

\begin{table}[!htbp]
    \centering
    \scriptsize
    \setlength{\tabcolsep}{2pt}
    \caption{Supplementary ablation results of EM (\%) and F1 (\%) under different builders on HotpotQA.}
    \label{tab:ablation_supp}
    \renewcommand{\csragrowstretch}{0.94}
    \resizebox{\columnwidth}{!}{
    \begin{tabular}{lcccccc}
        \toprule
        \multirow{2}{*}{Method} &
        \multicolumn{2}{c}{\textbf{GPT-4o-mini}} &
        \multicolumn{2}{c}{\textbf{Llama-3.1-70B}} &
        \multicolumn{2}{c}{\textbf{Qwen-3-14B}} \\
        \cmidrule(lr){2-3} \cmidrule(lr){4-5} \cmidrule(lr){6-7}
        & EM & F1 & EM & F1 & EM & F1 \\
        \midrule
        CS-RAG & 54.3 & 69.8 & 53.4 & 68.9 & 52.7 & 68.0 \\
        \midrule
        w/o Intersection & 52.5 & 67.4 & 51.2 & 66.0 & 50.3 & 64.7 \\
        w/o Typed-Var & 53.0 & 68.2 & 51.9 & 67.0 & 51.0 & 65.8 \\
        w/o Rerank (Con) & 51.8 & 66.6 & 50.2 & 64.7 & 49.0 & 63.2 \\
        w/o Rerank (Txt) & 51.3 & 65.9 & 49.5 & 63.8 & 48.2 & 61.9 \\
        \bottomrule
    \end{tabular}
    }
\end{table}

\begin{figure*}[t]
    \centering
    \includegraphics[width=1.00\linewidth]{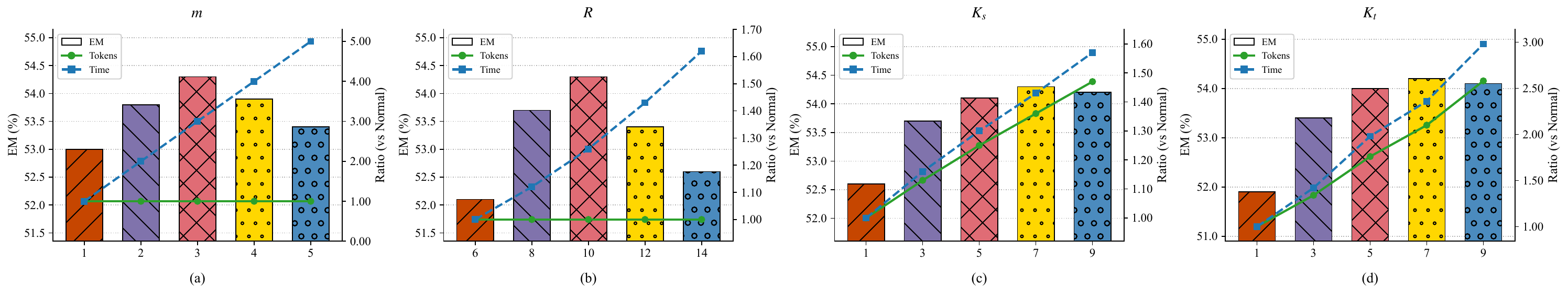}
    \caption{Hyperparameter sensitivity on HotpotQA with KGs constructed by GPT-4o-mini. We vary (a) the number of relation variants $m$, (b) the relation-aligned candidate budget $R$, (c) the structural candidate budget $K_s$, and (d) the textual recovery budget $K_t$. Results include EM, normalized retrieval time, and normalized retrieved evidence tokens.}
    \label{fig:hparam}
\end{figure*}

We further ablate several design choices in binding propagation and reranking on HotpotQA. 
w/o Intersection replaces intersection-based binding consolidation with union, and w/o Typed-Var removes type information in variables.
w/o Rerank (Con) and w/o Rerank (Txt) remove cross-encoder reranking in constraint-based retrieval and textual recovery, respectively. 
Table~\ref{tab:ablation_supp} shows that w/o Intersection reduces QA performance,
supporting the use of agreement across constraints to filter candidate bindings before propagation.
Removing type information in variables causes a smaller but consistent drop, suggesting that type information helps prevent variables with similar surface semantics from being mixed during intermediate binding.
Without contextual reranking in constraint-based retrieval, structurally adjacent but semantically weak triples are more likely to remain in the candidate pool. Without reranking in textual recovery, the fallback branch may pass weakly relevant passages into the final evidence context. The larger drops of these two variants suggest that reranking improves accuracy and stabilizes retrieval when KG quality varies across builders.

\subsection{Supplementary of Hyperparameter details.}
\label{appsec:hparam_supp}

Figure~\ref{fig:hparam} reports the results of $m$, $R$, $K_s$, and $K_t$ on HotpotQA with KGs constructed by GPT-4o-mini. Retrieval time and retrieved evidence tokens are normalized by the smallest setting of each hyperparameter.
Moderate values of $m$ perform better, as too few relation variants may miss valid structural matches, while too many expand the matching space and introduce loosely related candidates. Increasing the relation-aligned candidate budget $R$ initially improves EM, but further increases introduce more compatible yet irrelevant candidates into contextual reranking, lowering EM at a higher retrieval cost.
The token ratio remains unchanged for $m$ and $R$, as they affect candidate search and reranking before final evidence packing without changing the number of evidence items fed into the final QA model. Preserving more structural candidates with $K_s$ initially improves EM, but the gain quickly saturates once the main constraint-compatible evidence is covered. Larger textual recovery budgets ($K_t$) increase the token ratio more clearly, as recovered passages are directly included in the final QA model.
These trends suggest that the parameters in both retrieval and recovery should be set to moderate values, as larger budgets mainly increase retrieval time or retrieved evidence token counts after the useful evidence has been covered.

\setlength{\textfloatsep}{16pt plus 2pt minus 2pt}
\setlength{\intextsep}{9pt plus 1pt minus 1pt}
\captionsetup[table]{skip=6pt}
\subsection{Constraint Statistics}
\label{appsec:decomp_stats}

\begin{table}[!htbp]
\centering
\small
\setlength{\tabcolsep}{3pt}
\caption{Constraint statistics across three datasets.}
\label{tab:decomp_stats}
\renewcommand{\csragrowstretch}{0.90}
\begin{tabular*}{\columnwidth}{@{\extracolsep{\fill}}lccc@{}}
\toprule
Metric & MuSiQue & 2Wiki & HotpotQA \\
\midrule
Avg. constraints & 2.89 & 2.41 & 2.49 \\
Avg. variables & 2.32 & 2.35 & 1.65 \\
One-variable (\%) & 57.08 & 60.95 & 70.24 \\
Two-variable (\%) & 40.23 & 38.59 & 27.78 \\
\bottomrule
\end{tabular*}
\end{table}

Table~\ref{tab:decomp_stats} reports the average numbers of atomic constraints and variables, as well as the proportions of one-variable and two-variable constraints, to examine whether CS-RAG represents multi-hop queries through explicit variable dependencies beyond shallow one-hop expansions. The proportions are computed over all constraints in $\mathcal{P}$.
MuSiQue has the largest average number of constraints and a high average number of variables, reflecting its more compositional query structure and the need to ground intermediate entities across longer constraint chains.
2Wiki has fewer constraints on average, but its average number of variables is comparable to and slightly higher than MuSiQue's, suggesting compact but variable-dense constraint plans where multiple entities or values must be resolved.
In contrast, HotpotQA has a lower average number of variables and a higher proportion of one-variable constraints, suggesting that more constraints can be directly executed from observed anchors, with fewer variable-variable dependencies.
The remaining constraints belong to $\mathcal{P}_{\mathrm{obs}}$. Most simply restate query conditions without introducing new binding targets.
\subsection[Empirical Validation of the Sufficiency Check]{Empirical Validation of the Sufficiency Check}
\label{app:sufficiency_correctness}

\begingroup
\begin{table}[!ht]
\centering
\scriptsize
\setlength{\tabcolsep}{2pt}
\caption{Resolved-hop rates (Res., \%) and supported-binding rates (SBR, \%) across KG builders on all evaluation queries of three datasets. Res. denotes the percentage of resolved hops.}
\label{tab:builder_sufficiency}
\renewcommand{\csragrowstretch}{1.04}
\pgfmathsetlengthmacro{\sufficiencytablewidth}{0.98\columnwidth / \mainbuildertablescale}
\scalebox{\mainbuildertablescale}{%
\begin{tabular*}{\sufficiencytablewidth}{@{\extracolsep{\fill}}lcccccc@{}}
\toprule
\multirow{2}{*}{KG Builder}
& \multicolumn{2}{c}{MuSiQue}
& \multicolumn{2}{c}{2Wiki}
& \multicolumn{2}{c}{HotpotQA} \\
\cmidrule(lr){2-3} \cmidrule(lr){4-5} \cmidrule(lr){6-7}
& Res. & SBR
& Res. & SBR
& Res. & SBR \\
\midrule
Qwen-3-1.7B   & 13.8 & 91.4 & 18.4 & 95.2 & 15.6 & 88.8 \\
Qwen-3-14B    & 26.5 & 87.6 & 32.7 & 92.4 & 38.9 & 86.5 \\
Llama-3.1-70B & 30.4 & 84.8 & 38.3 & 90.7 & 52.7 & 87.8 \\
GPT-4o-mini   & 44.9 & 88.0 & 45.8 & 93.6 & 61.5 & 90.3 \\
Gemini-2.5-flash-lite   & 46.9 & 82.5 & 42.9 & 87.4 & 60.4 & 91.3 \\
\bottomrule
\end{tabular*}%
}
\end{table}
\endgroup

To validate the sufficiency check as an empirical signal for structural propagation across KG builders, we run CS-RAG on all evaluation queries in each dataset. We use the same trade-off thresholds as in the main experiments, with \(\gamma=1.5\) for MuSiQue, \(\gamma=1.6\) for 2Wiki, and \(\gamma=2.0\) for HotpotQA, keeping each threshold fixed across builders.
For each hop evaluated by the sufficiency check, we record its status after binding consolidation and use GPT-4o-mini to judge whether its consolidated binding set contains at least one binding supported by the gold evidence for the constraint. Res. denotes the percentage of evaluated hops that remain resolved after consolidation, and SBR denotes the percentage of these resolved hops with at least one supported binding.

Table~\ref{tab:builder_sufficiency}
shows that stronger builders generally yield more resolved hops after sufficiency checking and binding consolidation, allowing structural retrieval to proceed on more constraints.
For weaker builders, Res. decreases substantially, while SBR remains consistently high across builders and datasets.
High SBR indicates that
consolidated binding sets of resolved hops
usually contain at least one evidence-supported binding and thus offer a viable starting point for subsequent retrieval,
supporting the use of $N_{\mathrm{eff}}$ as an empirical signal for structural propagation. The remaining unsupported cases may arise from singleton pools or concentrated scores over incorrect candidates, which can pass the check without a valid binding. With thresholds fixed across builders, this pattern suggests that the check adapts the extent of structural propagation to the available KG evidence, reserving more hops for textual recovery when structural support is limited.

\section{Others}
\label{app:others}
Due to space limitations, additional materials, including prompts and case studies, will be provided in our code repository.

\end{document}